\documentclass{amsart}

\usepackage{color}

\newtheorem{theorem}{Theorem}[section]

\newtheorem{corollary}{Corollary}

\theoremstyle{conjecture}

\theoremstyle{definition}
\newtheorem{definition}{Definition}
\theoremstyle{question}

\theoremstyle{questions}

\theoremstyle{remark}
\newtheorem*{remark}{Remark}
\theoremstyle{remarks}

\theoremstyle{example}
\newtheorem*{example}{Examples}
\numberwithin{equation}{section}

\begin{document}
\title{On a Spherical Analogue of the Harmonic Oscillator}

\author{Van Higgs}
\email{van.higgs@mail.utoronto.ca}

\author{Doug Pickrell}
\email{pickrell@arizona.edu}

\begin{abstract} A linear quantum harmonic oscillator factors into one dimensional oscillators and can be solved using creation and annihilation operators. We consider a spherical analogue. This analogue does not factor. The two dimensional case is critical, and we compute the spectrum and partition function. This is of interest because the spherical oscillator is potentially relevant to chiral models in 2d quantum field theory. 
\end{abstract}
\maketitle

\setcounter{section}{-1}

\section{ Introduction}

A d-dimensional quantum harmonic oscillator is characterized as the unique $SO(d)$ invariant self-adjoint second order differential operator on $\mathbb R^d$ having principal symbol the Euclidean metric, a Gaussian ground state, and (in a quantum field theoretic context) vanishing ground state energy. The oscillator factors as a sum of independent one dimensional oscillators, and it can be solved using annihilation and creation operators; this is true also of q-deformations of the harmonic oscillator (see e.g. \cite{EM}). 

Consider the stereographic correspondence
$$S^{d} \leftrightarrow \mathbb R^d \cup\{\infty\}: (x_1,...,x_{d},y) \leftrightarrow x:=\frac{(x_1,...,x_d)}{1-y}$$
The round metric in stereographic coordinates is given by the formula
$$(ds)_{S^d}=\frac{2}{1+|x|^2}(ds)_{\mathbb R^d} $$
the volume form is given by
$$dV_{S^d}=2^d (1+|x|^2)^{-d}dV_{\mathbb R^d}=2^d\frac{r^{d-1}}{(1+r^2)^{d}}dr dV_{S^{d-1}} $$
and the nonnegative Laplace operator is given by
$$\Delta_{S^d}=-\frac 14 (1+r^2)^2\left((\frac{d}{dr})^2+\frac 1r(1+(d-2)\frac{1-r^2}{1+r^2})\frac{d}{dr}+\frac {1}{r^2}\Delta_{S^{d-1}}\right) $$
where $x\in \mathbb R^d$ and $r=|x|$. 

Given $\omega \ge 0$, our spherical analogue of the harmonic oscillator is the differential operator
$$L_{\omega}:=\Delta_{S^d}+\omega^2r^2$$
This is essentially self-adjoint on compactly supported smooth functions on $\mathbb R^d$.
It is routine to find the spectrum of this operator, it does not factor, and one cannot use annihilation and creation operators. 
This begs the question of why it is of special interest.

A Gaussian is the ground state of a linear harmonic oscillator. Consider the Cauchy type density $(1+r^2)^{-\omega}$, relative to $dV_{S^d}$,
in stereographic coordinates. Is this density the ground state of
a natural geometric operator? If $d=2$, then the answer is yes, there is a unique $SO(d)$ invariant self-adjoint second order differential operator having principal symbol determined by the rotationally invariant metric for the d-sphere
and having the ground state $(1+r^2)^{-\omega}$, with vanishing ground state energy:
$$\mathbf L_{\omega}:=L_{\omega}-\omega=\Delta_{S^d}+\omega^2r^2-\omega$$

The story is essentially the same, although slightly more complicated in an interesting way, when we consider the dimension $d$ as a parameter.
In general the quadratic potential is quite singular from the point of view of the original $d$-sphere, and breaks the $O(d)$ symmetry down to $O(d-1)$. 
The domain of the operator depends sensitively on $\omega$, e.g. $1$ is in the domain when $\omega=0$, and not when $\omega>0$.
It turns out that for the radial eigenfunction equation $r=0$ is regular when $d=1$ and singular otherwise. This is one way to anticipate that various quantities 
depend discontinuously on the dimension parameter $d$, and $d=2$ is (predictably) critical. 

When $d=2$, this Cauchy type ground state arises naturally in connection
with the two dimensional quantum sigma model with target a compact semisimple Lie group, i.e. the principal chiral model. It is not known how to rigorously formulate
this model, although there is a conjecture for the scattering theory (see \cite{ORW}). It is possible that the Hamiltonian in finite volume, i.e. when space is
a circle $S^1$, factors as a sum of ordinary harmonic oscillators and spherical oscillators. We hope
to use the partition functions associated to spherical oscillators as a means to test this possibility.

\subsection{Plan of the Paper}

In Section \ref{Self-Adjointness} we recall standard arguments which show that $L_{\omega}$ (with appropriate domain) is a nonnegative self-adjoint operator
with compact resolvent (hence discrete spectrum).

In Section \ref{GroundState} we find the ground state and observe how various formulas depend discontinuously on
the dimension parameter; $d=2$ is critical.

In Sections \ref{spectrum2} and \ref{spectrumgen}  we find the spectrum. We have separated out the $d=2$ case because of its importance and relative simplicity.

In Section \ref{CM} we discuss the potential relevance of our spherical oscillator to the two dimensional principal chiral model.

\section{Self-Adjointness}\label{Self-Adjointness}

Recall the following three examples from page 331 of \cite{Rudin} of unbounded operators $(T_k,D(T_k))$ in the real Hilbert
space $H:=L^2([0,1],dx;\mathbb R)$: $D(T_1)$ consists of absolutely continuous functions $f$ on $[0,1]$ such that $f'\in L^2$,
$$D(T_2):=\{f\in D(T_1):f(0)=f(1)\}, \text{  and  } D(T_3):=\{f\in D(T_1):f(0)=f(1)=0\} $$
and in each case $T_kf=\frac 1i f'$. Then
$$T_1^*=T_3, \qquad T_2^*=T_2, \qquad T_3^*=T_1 $$ 
Heuristically the spectrum of $T_1$ consists of all $\lambda\in\mathbb C$ with corresponding eigenfunction $f(x):=e^{i\lambda x}$,
the spectrum of the self-adjoint operator $T_2$ is $2\pi\mathbb Z$, and the spectrum of $T_3$ is empty.
It is important that our spherical oscillators are self-adjoint, to avoid the "pathology" that occurs in the first and third examples.

An efficient way to define a non-negative self-adjoint operator is to use symmetric forms, see page 11 of \cite{Faris}.
The Hilbert space is $H:=L^2(dV_{S^d}, \mathbb R)$. The form is
$$Q_{\omega}(f,g):=\int (df\wedge *dg+\eta^2r^2fg dV_{S^d})$$
where $*$ denotes the star operator for $S^d$ (which is conformally invariant precisely when $d=2$),
and $Q=Q_{\omega}$ is defined for $f,g$ in the dense subspace $E\subset H$ for which $(Qf,f)$ and $Q(g,g)$ are finite.

\begin{definition}(following \cite{Faris}) 
$$D(L_{\omega}):= \{g\in E: f \to Q(f,g) \text{ is continuous on }H\}$$
and $Q(f,g)=\langle f,L_{\omega}g\rangle_{H}$.
\end{definition}

\begin{theorem} $L_{\omega}$ is non-negative self-adjoint and the resolvent is compact (more precisely
in the Dixmier Schatten class $\mathcal L_{d-1}^+$). \end{theorem}

The first statement follows from the Form Representation Theorem 2.1, page 12 of \cite{Faris}. 
Note that $L_{\omega}$ as a differential operator is elliptic in $\mathbb R^d$ and equal to the self-adjoint
$L_{\omega}$ on sufficiently differentiable functions. The second statement follows from the fact
that the inverses are pseudo-differential operators of order -2. Note that in the case $d=2$ the inverses 
are actually in the Dixmier Schatten class $\mathcal L_1^+$.

\section{The Ground State}\label{GroundState}

The following shows how the ground state energy depends discontinuously on $\omega$ and the dimension parameter.

\begin{theorem} If $\omega>0$, then for any dimension $d$, the ground state for $L_{\omega}$
 is
$$(1+r^2)^{-\frac{((d-2)^2+16\omega^2)^{1/2}-(d-2)}{4}}$$
and the ground state energy is
$$ d\frac{((d-2)^2+16\omega^2)^{1/2}-(d-2)}{8}$$
where we choose the positive square root in both formulas.

If $\omega=0$, then the above formulas are valid, provided that for $d\le 2$, we choose the negative
 square roots, and for $d\ge 2$, we choose the positive square roots.
\end{theorem}

\begin{proof}Define $\delta(r):=(1+r^2)^{-\eta}$, $\eta>0$. Then
$$L_{\omega}\delta=\frac{(-2\eta^2+(d-2)\eta-2\omega^2)r^2+d\eta}{2}\delta$$
For $\delta$ to be an eigenfunction (with eigenvalue $d\eta/2$) it is necessary that
$2\eta^2+(d-2)\eta-2\omega^2=0$, i.e.
$$\eta=\frac{-(d-2)\pm ((d-2)^2+16\omega^2)^{1/2}}{4} $$
This must be nonnegative, vanish when $\omega=0$, and
$\delta$ must also be in the domain, in particular
$$\int_0^{\infty} \delta(r)^2 (1+r^2)^{-d}r^{d-1}dr<\infty $$
These conditions determine the signs of the square roots that appear in the statement of
the theorem.
\end{proof}

\begin{remark} If $d=2$, let $w=1/\overline{z}$, and in general let $w$ denote the reflection of 
$x$ through the $d-1$ sphere, which we view as a coordinate for $\infty$. In this coordinate the ground state equals
$$\frac{|w|^{2\eta}}{(1+|w|^2)^{\eta}}$$
where $\eta=\frac{((d-2)^2+16\omega^2)^{1/2}-(d-2)}{4}$. This vanishes at $\infty$ in the sphere. It is smooth
when $\eta$ is integral. Otherwise it has Lipschitz type behavior.
\end{remark} 

As motivation (see chapter 2 of \cite{BGV} for the numerous definitions) recall that a Laplace type operator is determined by its principal symbol (the highest order terms), a connection (essentially the first order terms),
and a potential. In our narrow context, the latter two conditions can be replaced with invariance, self-adjointness and knowledge of the ground state.

\begin{theorem} Suppose $\omega>0$. There is a unique $SO(d)$ invariant self-adjoint second order differential operator having principal symbol determined by the standard metric for the d-sphere
and having the ground state $(1+r^2)^{-\omega}$, with vanishing ground state energy,

$$\mathbf L_{\omega}=-\frac 14 (1+r^2)^2\left((\frac{d}{dr})^2+\frac 1r(1+(d-2)\frac{1-r^2}{1+r^2})\frac{d}{dr}+\frac {1}{r^2}\Delta_{S^{d-1}}\right)$$
$$-d\frac{((d-2)^2+16\omega^2)^{1/2}-(d-2)}{8}$$

\end{theorem}

\begin{proof} We are led to this formula by the calculation in the preceding proof (with $\eta=\omega$).

Imagine adding lower order terms to $L_{\omega}$. The corresponding quadratic form, in terms of the radial parameter $r$ (we can ignore $\theta$ because of rotational invariance), would (because of self-adjointness) be of the form
$$\int_{r=0}^{\infty} b(r)(f'(r)g(r)+f(r)g'(r))\frac{r}{(1+r^2)^2} dr$$
$$=\int_{r=0}^{\infty} (b(r)f'(r)g(r)\frac{r}{(1+r^2)^2}-\frac{\partial }{\partial r}(b(r)f(r)\frac{r}{(1+r^2)^2}) dr$$
This calculation shows that for the differential operator, the derivative terms cancel, and this only adds zeroth order terms. The fact that we know the ground state implies that $L$ is determined up to a constant. The fact the ground state energy vanishes determines the constant.
 \end{proof}

\section{Spectrum: The Two Dimensional Case}\label{spectrum2}

In this section $d=2$.

\begin{theorem}\label{spectrum22} Suppose $\omega\ge 0$.

(a) The ground state for $L_{\omega}$ is
$$(1+r^2)^{-\omega}$$
with corresponding eigenvalue $\omega$.

(b) The spectrum is  
$$\lambda_{m,n}=(m+\frac 12(n+1+\sqrt{n^2+4\omega^2}))^2-\omega^2-\frac 14 $$
$m,n\ge 0$, with multiplicity $\mathbf m(m,n)=1$ when $n=0$ and $\mathbf m(m,n)=2$ otherwise (The values of the $\lambda_{m,n}$ are not necessarily 
distinct, see part (c) and the examples in Remark \ref{clarify} below). The corresponding eigenfunctions are of the form
$$ \frac{r^n}{(1+r^2)^{(\sqrt{4\omega^2+4\lambda+1}-1)/2}}F(-m-\sqrt{4\omega^2+n^2},-m,n+1,-r^2) f(\theta)$$
where $F$ is the hypergeometric function (an even $2m$ degree polynomial in $r$ in our case) and $f(\theta)$ is a linear combination of $cos(n\theta)$ and $sin(n\theta)$.

(c) If $\omega=0$, then the spectrum is $N(N+1)$ with multiplicity $2N+1$, $N=0,1,...$. 

\end{theorem}

\begin{remark}\label{clarify} In reference to part (b), note that if $\omega=2\sqrt{3}$, then $\lambda_{3,1}=44=\lambda_{1,4} $, and if $\omega = 4$, 
then $\lambda_{9, 6}=\lambda_{1, 15}$. Thus degeneracy
of the eigenvalues can definitely occur, even if $\omega\ne 0$. Nonetheless in part (b) we are counting the eigenvalues correctly.
When $\omega=2\sqrt{3}$, then $\lambda_{3,1}$ and $\lambda_{1,4} $ are the same, but if we attach a multiplicity $2$ to each of them, it is
the same as attaching multiplicity $4$ to the common value. 

We have not been able to find a robust deterministic condition which rules out degeneracy.

\end{remark}

\begin{proof} We have already proven (a).

To find eigenvalues and eigenfunctions, we use separation of variables (we will recall the standard justification at the end of the proof). Suppose we consider an eigenfunction of the form
$g(r)f(\theta)$. Then
$$-\frac 14(1+r^2)^2\left(g''(r)f(\theta)+\frac 1r g'(r)f(\theta)+\frac{1}{r^2}g(r)f''(\theta)\right)+\omega^2 r^2g(r)f(\theta)=\lambda g(r)f(\theta)$$
Then
$$-(1+r^2)^2\left(\frac{g''(r)}{g(r)}+\frac 1r \frac{g'(r)}{g(r)}+\frac{1}{r^2}\frac{f''(\theta)}{f(\theta)}\right)+4\omega^2 r^2=4\lambda $$
and
$$\frac{r^2}{(1+r^2)^2}\left(-(1+r^2)^2(\frac{g''(r)}{g(r)}+\frac 1r \frac{g'(r)}{g(r)})+4\omega^2 r^2-4\lambda\right) =\frac{f''(\theta)}{f(\theta)}$$
This implies that $f(\theta)$ is a combination of $cos(n\theta)$ and $sin(n\theta)$, where $n$ is a nonnegative integer, and
\begin{equation}\label{radialeigenfunction}r^2g''(r)+r g'(r)+\left(-n^2+4(1+r^2)^{-2}(-\omega^2 r^4+\lambda r^2)\right)g(r)=0 \end{equation}
This has regular singular points and is therefore equivalent to the hypergeometric equation. The indicial equation
is $\alpha^2-n^2=0$ and $g$ has the form
$$g(r)=r^{\pm n}\sum_{l=0}^{\infty}g_lr^l $$
The relevant solution involves $+n$, so that $g$ is regular at $r=0$. One finds (using Maple)
\begin{equation}\label{d=2eigenfn}g(r)=\frac{r^n}{(1+r^2)^{(\sqrt{4\omega^2+4\lambda+1}-1)/2}}F(a,b,n+1,-r^2)\end{equation}
where $F$ denotes the hypergeometric function with
$$a=\frac 12(-\sqrt{4\omega^2+4\lambda+1}+n+1-\sqrt{4\omega^2+n^2})$$
and
$$b=\frac 12(-\sqrt{4\omega^2+4\lambda+1}+n+1+\sqrt{4\omega^2+n^2})=a+\sqrt{4\omega^2+n^2}$$

A necessary condition for $\lambda$ to be an eigenvalue is that $g(r)$ is square integrable on $(0,\infty)$ with respect
to $(1+r^2)^{-2}rdr$. Thus $a$ or $b$ must be a non-positive integer. As above $a<b$.

\begin{example}\label{example1} Suppose $\omega=0$. In this case $b=a+n$. Suppose $b=-m$, a non-positive integer. In this case
$\lambda=\lambda_{m,n}=(m+n)(m+n+1)$. Square integrability holds. Now consider the possibility that $a=-m$, a non-positive integer and $b=a+n=-m+n>0$, i.e. $n>m$. In this case $\lambda=m(m+1)$. For $g(r)$ to be square integrable it is necessary that
$$2(\sqrt{1+4\lambda}-1)-2n-4m+3>1$$
or
$$1+4\lambda=1+4m(m+1)>(n+2m)^2=n^2+4mn+4m^2$$
or
$$1+4m(1)>(n+2m)^2=n^2+4mn$$
Square integrability fails. So we need both a and b to be non-positive.
If we fix $\lambda=N(N+1)$, then the multiplicity is $2N+1$. For if $\lambda=\lambda_{m,n}$, when $n=0$ we can only choose $cos(n\theta)=1$
and when $n>0$ we can choose $cos(n\theta)$ or $sin(n\theta)$.
\end{example}

Now suppose $\omega>0$.

Suppose that $b=-m$, a non-positive integer. Since $a<b$, $F(a,b,c,-r^2)$ is a polynomial of degree $m$ in $z=-r^2$. In this case
(after a remarkable simplification)
$$\lambda:=\lambda_{m,n}=(m+\frac 12(n+1+\sqrt{n^2+4\omega^2}))^2-\omega^2-\frac 14 $$
as in the statement of the Theorem.
We first find the condition for $g\in H:=L^2((1+r^2)^{-2}rdr)$. Since $F(a,b,n+1,-r^2)$ is a polynomial of degree $2m$
in $r$, $g\in L^2$ if and only if (using the expression (\ref{d=2eigenfn}) for the eigenfunction)
$$2(\sqrt{4\omega^2+4\lambda+1}-1)-2n-4m)+3>1$$
or
$$4\omega^2+4\left((m+\frac 12(n+1+\sqrt{n^2+4\omega^2}))^2-\omega^2-\frac 14\right)+1>(n+2m)^2 $$
or
$$4(m+\frac 12(n+1+\sqrt{n^2+4\omega^2}))^2>(n+2m)^2 $$
This condition is always satisfied, so $g(r)f(\theta)$ is in the Hilbert space $H$. To see that 
it is actually in $D(L_{\omega})$, for $F\in C^{\infty}(\mathbb R^2)$ having compact support, we use
integration by parts. This is justified, because the functions are smooth in the plane and $F$ vanishes outside of a compact set.
Consequently 
$$\int (dF\wedge *d(gf)+\omega^2r^2Fgf)=\int F((\Delta+\omega^wr^2)(gf))dV_{S^2}=\langle F,\lambda gf\rangle_{H}$$
This shows that the function $H \to \mathbb R: F \to Q(F,gf)$ is continuous, hence $gf\in D(L_{\omega})$. Thus $\lambda_{m,n}$ is actually
an eigenvalue.

Suppose that $a=-m$ and $b=a+\sqrt{4\omega^2+n^2}>0$. In this case
$$\lambda=\lambda_{m,n}'=(m+\frac 12(n+1-\sqrt{n^2+4\omega^2}))^2-\omega^2-\frac 14 $$
(the sign of the square root has changed).
Again $F(a,b,n+1,-r^2)$ is a polynomial of degree $2m$
in $r$. Thus $g\in L^2$ if and only if
$$2(\sqrt{4\omega^2+4\lambda+1}-1)-2n-4m+3>1$$
or
$$4\omega^2+4\left((m+\frac 12(n+1-\sqrt{n^2+4\omega^2}))^2-\omega^2-\frac 14\right)+1>(n+2m)^2 $$
or
$$4(m+\frac 12(n+1-\sqrt{n^2+4\omega^2}))^2>(n+2m)^2 $$
This implies $1>\sqrt{n^2+4\omega^2}$, hence $n=0$.
But then
$b=-m+2\omega\le 0$, which contradicts our assumption that $b>0$. Thus
$\lambda_{m,n}'$ is not an eigenvalue.

We briefly recall why separation of variables yields a complete set of eigenfunctions. 
Suppose that $\psi$ is an eigenfunction corresponding to an eigenvalue $\lambda$. By ellipticity $\psi$ is smooth on $\mathbb R^2$ and will have
an expansion of the form
$$\psi(r,\theta)=\sum_{n=1}^{\infty} (\psi_n(r)cos(n\theta)+\widetilde{\psi_n}(r)sin(n\theta) $$
This implies
$$L_\omega\psi=\sum_{n=1}^{\infty} \left((-\frac 14(1+r^2)^2\left(\psi_n''(r)+\frac 1r \psi_n'(r)-\frac{1}{r^2}n^2\psi_n(r)\right)+\omega^2 r^2\psi_n(r)\right)cos(n\theta)$$
plus a similar sum involving $\widetilde{\psi}_n$ and sine functions. This equals $\lambda \psi$. This implies that for each $n$
$$(-\frac 14(1+r^2)^2\left(\psi_n''(r)+\frac 1r \psi_n'(r)-\frac{1}{r^2}n^2\psi_n(r)\right)+\omega^2 r^2\psi_n(r)=\lambda\psi_n(r)$$
and similarly for the $\widetilde{\psi}_n$. Since we found all of the solutions for the radial eigenfunction problem (\ref{radialeigenfunction}),
this implies that $\lambda$ is equal to one of the $\lambda_{m,n}$. 

This completes the proof. 

\end{proof}

\section{Spectrum: The General Case}\label{spectrumgen}

If $\omega=0$, then it is well-known that the spectrum of the positive Laplacian on the $d$-sphere is $N(N+d-1)$ with multiplicity 
$$ \left(\begin{matrix}(N+d\\N\end{matrix}\right)-\left(\begin{matrix}(N+d-2\\d\end{matrix}\right), \qquad N=0,1,...$$ 
In particular the ground state energy is zero.

\begin{theorem}\label{spectrumgen2}Suppose that $\omega>0$. For $L_{\omega}$\\

(a) The ground state is
$$\delta(r):=(1+r^2)^{-\frac{|((d-2)^2+16\omega^2)^{1/2}|-(d-2)}{4}}$$
with ground state energy 
$$ d\frac{|((d-2)^2+16\omega^2)^{1/2}|-(d-2)}{8}$$

(b) If $d\ge 2$, then the spectrum is 
$$\lambda_{m,n}=(m + \frac 14(|((2n+d - 2)^2 + 16\omega^2)^{1/2}| + (d - 2 + 2(n + 1)))^2 - \omega^2 - \frac{1}{4}(d - 1)^2 $$
$m,n\ge 0$, with multiplicity $\mathbf m(m,n)=1$ when $n=0$ and   
$$\mathbf m(m,n)=\left(\begin{matrix}N+d-1\\N\end{matrix}\right)-\left(\begin{matrix}N+d-3\\N-3\end{matrix}\right)$$
when $n>0$, where $N=n(n+d-2)$ (As in the $d=2$ case, these values are not necessarily distinct, which happens for example when $\omega=0$).

(b)' If $d=1$, then the spectrum is $\lambda_m:=\lambda_{m,n=0}$, $m\ge 0$, where we choose the negative square root in the formula for $\lambda_{m,n=0}$:
$$ \lambda_m=(m + \frac 14(1-|(1 + 16\omega^2)^{1/2}|))^2 - \omega^2, \qquad m=0,1,... $$
If $\omega>0$, the spectrum is multiplicity free.  

\end{theorem}

\begin{proof} We proved (a) in Section \ref{GroundState} and observed this depends discontinuously on $\omega$ and the dimension.

(b) Assume $d>1$. To find eigenvalues and eigenfunctions, we use separation of variables. Suppose we consider an eigenfunction of the form
$g(r)f(\theta)$, where (abusing notation) we are writing $f(\theta)$ for a function that is defined on the $d-1$ sphere. Then
$$-\frac 14(1+r^2)^2\left(g''(r)f(\theta)+\frac{1}{r}(1+(d-2)\frac{1-r^2}{1+r^2}) g'(r)f(\theta)+\frac{1}{r^2}g(r)\Delta_{S^{d-1}}f(\theta)\right)$$
$$+\omega^2 r^2g(r)f(\theta)=\lambda g(r)f(\theta)$$
Thus
$$-(1+r^2)^2\left(\frac{g''(r)}{g(r)}+\frac{1}{r}(1+(d-2)\frac{1-r^2}{1+r^2})\frac{g'(r)}{g(r)}+\frac{1}{r^2}\frac{\Delta_{S^{d-1}}f(\theta)}{f(\theta)}\right)+4\omega^2 r^2=4\lambda $$
and
$$\frac{r^2}{(1+r^2)^2}\left(-(1+r^2)^2(\frac{g''(r)}{g(r)}+\frac{1}{r}(1+(d-2)\frac{1-r^2}{1+r^2}) \frac{g'(r)}{g(r)})+4\omega^2 r^2-4\lambda\right) =\frac{\Delta_{S^{d-1}}f(\theta)}{f(\theta)}$$
This implies that $f(\theta)$ is an eigenfunction for  $\Delta_{S^{d-1}}$. The corresponding eigenvalue is
of the form $n(n+d-2)$, where $n$ is a nonnegative integer, and
$$r^2g''(r)+(1+(d-2)\frac{1-r^2}{1+r^2})r g'(r)+\left(-n(n+d-2)+4(1+r^2)^{-2}(-\omega^2 r^4+\lambda r^2)\right)g(r)=0$$
This has regular singular points and is equivalent to the hypergeometric equation. At $r=0$ the indicial equation is
$$\alpha^2+(d-2)\alpha-n(n+d-2)=0$$
for which the roots are $\alpha=n$ and $\alpha=-(n+d-2)$.
For $g(r)$ to be regular at $r=0$, $g(r)$ must be of the form
$$g(r)=r^{n}\sum_{l=0}^{\infty}g_lr^l $$
Using Maple,  
\begin{equation}\label{eigenfn}g(r)= r^{n} \left(1+r^{2}\right)^{\frac{d-1-\sqrt{(d-1)^2+4 \omega^{2}+4 \lambda }}{2}} F(a, b, n+\frac{d}{2}, -r^2)\end{equation}
where $F$ denotes the hypergeometric function, we choose the positive square root,  and
$$a=-\frac{\sqrt{(d-1)^2+4 \omega^{2}+4 \lambda }}{2}+\frac{n}{2}+\frac{d}{4}-\frac{\sqrt{(2n+d-2)^2+16 \omega^{2}}}{4}
$$
$$b=-\frac{\sqrt{(d-1)^2+4 \omega^{2}+4 \lambda }}{2}+\frac{n}{2}+\frac{d}{4}+\frac{\sqrt{(2n+d-2)^2+16\omega^{2}}}{4}
$$ 

For $\lambda$ to be an eigenvalue, $g(r)$ must be square integrable on $(0,\infty)$ with respect
to $(1+r^2)^{-d}r^{d-1}dr$. A necessary condition is that $a$ or $b$ is a non-positive integer.

Suppose that $a<b=-m$, a non-positive integer. In this case $F(a,b,n+\frac{d}{2},-r^2)$ is a polynomial of degree $m$ in $-r^2$. In this case
$\lambda=\lambda_{m,n}$ equals (after a somewhat miraculous simplification)
$$(m + \frac 14((2n+d - 2)^2 + 16\omega^2)^{1/2} + (d - 2 + 2(n + 1)))^2 - \omega^2 - \frac{1}{4}(d - 1)^2 $$

For $\lambda_{m,n}$ to be an eigenvalue it is necessary that $g\in L^2((1+r^2)^{-d}r^{d-1}dr)$. Now $F(a,b,n+\frac{d}{2},-r^2)$ is a polynomial of degree $2m$
in $r$. Thus $g\in L^2$ if and only if
$$2n+2(d-1-\sqrt{(d-1)^2+4\omega^2+4\lambda_{m,n}})+4m+d-1-2d<-1 $$
or
$$2n+4m+d-2<2\sqrt{(d-1)^2+4\omega^2+4\lambda_{m,n}} $$
or
$$(2n+4m+d-2)^2<4((d-1)^2+4\omega^2+4\lambda_{m,n}) $$
or
$$(2n+4m+d-2)^2<4(m + \frac 14(((2n+d - 2)^2 + 16\omega^2)^{1/2} + (d - 2 + 2(n + 1)))^2 $$
This condition is always satisfied, so $g(r)f(\theta)$ is in the Hilbert space. We use the same
integration by parts argument as in the $d=2$ case to see that $g(r)f(\theta)\in D(L_{\omega})$, which implies
$\lambda_{m,n}$ is an eigenvalue.

As in the $d=2$ case we rule out $\lambda$ as an eigenvalue
when $a=-m$ and $b>0$. 

Separation of variables yields a complete set of eigenfunctions by the same argument as in the proof
of the $d=2$ case.

Suppose that $d=1$. In this case we are not restricting $r$ to be positive. The eigenvalue equation
$$g''(r)+\frac{2r}{1+r^2}g'(r)+4(1+r^2)^{-2}(-\omega^2 r^2+\lambda )g(r)=0$$
no longer has a singularity at $r=0$, and
$$g(r)=c_1 LegendreP(\sqrt{16\omega^2 + 1}/2 - 1/2, 2\sqrt{\omega^2 + \lambda }, ir)$$
$$+c_2
LegendreQ(\sqrt{16\omega^2 + 1}/2 - 1/2, 2\sqrt{\omega^2 + \lambda }, ir) $$
where the signs of the square roots have to be determined.

For each $m=0,...$ we find that 
$$ LegendreQ(-|\sqrt{16\omega^2 + 1}|/2 - 1/2, -2|\sqrt{\omega^2 + \lambda_m }|, ir)$$
$$=C(1+r^2)^{-| \sqrt{\omega^2 + \lambda_m }|}
F(-m-\frac 12, -m,1-|\sqrt{\omega^2 + \lambda_m }|,ir)$$
is an eigenfunction corresponding to $\lambda_m=(m + \frac 14(1-|(1 + 16\omega^2)^{1/2}|)^2 - \omega^2$. Using
$$LegendreP(a,b,z)=C\left(\frac{z+1}{z-1}\right)^{b/2} F(-a,a+1,1-b,1/2-z/2))/((z-1)^(b/2)$$
which implies
$$|LegendreP(\sqrt{16\omega^2 + 1}/2 - 1/2, 2\sqrt{\omega^2 + \lambda_m }, ir)|$$
$$=C|F(-m-\frac 12, -m,1-|\sqrt{\omega^2 + \lambda_m }|,ir)$$
we can rule out that LegendreP is an eigenfunction, because this is not square integrable, unless $m=0$. 
Hence the spectrum is multiplicity free.
\end{proof}

\begin{corollary}\label{spectrum5} Suppose that $d\ge 1$, $\omega\ge 0$, and let $\mathbf L_{\omega}$ denote $L_{\omega}$ 
minus its ground state energy.\\

(a) The ground state is
$$\delta(r):=(1+r^2)^{-\frac{d-2+|((d-2)^2+16\omega^2)^{1/2}|}{4}}$$

(b) If $d\ge 2$, then the spectrum is 
$$\lambda_{m,n}=(m + \frac 14(|((2n+d - 2)^2 + 16\omega^2)^{1/2}| + (d - 2 + 2(n + 1)))^2 - ( \frac 14(|((d - 2)^2 + 16\omega^2)^{1/2}| + d ))^2 $$
$m,n\ge 0$, with multiplicity $\mathbf m(m,n)=1$ when $n=0$ and   
$$\mathbf m(m,n)=\left(\begin{matrix}(N+d-1\\N\end{matrix}\right)-\left(\begin{matrix}(N+d-3\\N-3\end{matrix}\right)$$
when $n>0$, where $N=n(n+d-2)$.

(b)' If $d=1$, then the spectrum is 
$$\lambda_{m}=(m + \frac 14(1-|(1 + 16\omega^2)^{1/2}|) )^2 - ( \frac 14(1-|(1 + 16\omega^2)^{1/2}|  ) )^2 $$
$m\ge 0$, with multiplicity $1$.

\end{corollary}

\section{Relation to the Principal Chiral Model}\label{CM}

In this section we will briefly outline why it is possible that the Hamiltonian for the quantum principal chiral model
might be a sum of ordinary harmonic oscillators and spherical harmonic oscillators. 

If $X$ is a Riemannian manifold, then the conformally invariant classical action for the sigma model with target $X$
on the cylinder $S^1\times \mathbb R$ is
\begin{equation}\label{1.1}S:W^{1,L^2}(S^1\times \mathbb R,X) \to \mathbb R: x \to \frac
12\int_{\Sigma}\left(\vert\frac {\partial x}{\partial t}\vert^2+\vert\frac {
\partial x}{\partial\theta}\vert^2\right)d\theta dt\end{equation}
The classical configuration space is the space of loops
$W^{1/2,L^2}(S^1,X)$ having the critical degree of smoothness for a loop, in the $L^2$ Sobolev sense. 

If $X=\mathbb R$, this is a massless free scalar field. In terms of the Fourier series expansion $x(t,\theta)=\sum_{n=-\infty}^{\infty} x_n(t)e^{in\theta}$,
the action is 
$$S(x)=\sum_{n=-\infty}^{\infty}\int_t(|x_n'(t)|^2+n^2|x_n(t)|^2)$$
which is quadratic. Consequently the Feynman measure for the free field is Gaussian. The ground state of the theory is obtained by taking the trace
on $S^1$. This is also a Gaussian measure, and it factors in terms of the Fourier modes $x_n$, $n=1,2,...$. From this we can deduce the form of the Hamiltonian
for the free field: it is a sum of ordinary harmonic oscillators (see pages 733-734 of \cite{G} for a lucid discussion).

Now suppose that $X$ is a simply connected compact Lie group with the bi-invariant metric; we assume $X=SU(2)$ (the 3-sphere) for simplicity. In this case
the classical action does not factor, it is necessary to use renormalization group techniques to even understand the model 
at an effective field theory level (see Lecture 3 of \cite{G}), and conformal invariance is broken and replaced by asymptotic freedom at the quantum level.  

How is the quantum model formulated? The state space for the model should be an $L^2$ space over a completion of the classical configuration space. 
It is remarkable that there exists a finite loop group bi-invariant measure (denoted by $\mu_0$)
which lives on a completion of the configuration space (the hyperfunction completion of the loop group; see \cite{DP}). This should determine the state space,
and (by doing a vacuum renormalization if necessary), one can assume $\mu_0$ is the ground state for the quantum model. The next key fact is that there is a 
noncommutative, or multiplicative, analogue
of Fourier series: Given a $g\in W^{1/2, L^2}(S^1,K)$, there is a unique `multiplicative Fourier series' (with $z=exp(i\theta)$)
\begin{equation}\label{RSF}g(z)=\prod_{i\ge 0}^{\rightarrow}\mathbf a(\eta_i)\left(\begin{matrix} 1&\overline{\eta}_iz^i\\
-\eta_iz^{-i}&1\end{matrix} \right)\left(\begin{matrix} e^{\chi(z)}&0\\
0&e^{-\chi(z)}\end{matrix} \right)\prod_{k\ge 1}^{\leftarrow}\mathbf a(\zeta_k)\left(\begin{matrix} 1&\zeta_kz^{-k}\\
-\overline{\zeta}_kz^k&1\end{matrix} \right)\end{equation}
where $\chi(z)=\sum\mathbf{\chi}_jz^j$ is a $i\mathbb R$-valued Fourier series
(modulo $2\pi i\mathbb Z$). Here the $\eta_i$ and $\zeta_k$ variables should be viewed as affine coordinates for 2-spheres. 

Although $ W^{1/2,L^2}(S^1,K)$ has measure zero with respect to $\mu_0$, the $\eta_i$, $\chi_j$ and $\zeta_k$
can be interpreted as random variables with respect to $\mu_0$, and conjecturally
$\mu_0$ factors:
\begin{equation}\label{productmeasure7}d\mu_0=\left(\prod_{i=0}^{\infty}\frac {1+2i}{\pi}\frac
{d\lambda (\eta_
i)}{(1+\vert\eta_i\vert^2)^{2+2i}}\right)\times \left(\prod_{
j=1}^{\infty}\frac
{4j}{\pi}e^{-4j\vert\mathbf{\chi}_j\vert^2}d\lambda
(\mathbf{\chi}_j)\right)$$
$$\times d\lambda
(e^{\chi_0})\times\left(\prod_{k=1}^{\infty}\frac
{2k-1}{\pi}\frac {d\lambda (\zeta_
k)}{(1+\vert\zeta_k\vert^2)^{2k}}\right),\end{equation}

This form of the vacuum, involving a product of Cauchy type densities and Gaussians, suggests an ansatz for the form of the Hamiltonian. The Gaussian density for $\chi$ (and asymptotic freedom) suggests that $\chi$ is a massless free field. The partition function can be expressed in terms of the $\eta$ function: 
$$trace(e^{-tH_{\chi}})=\frac{1}{\eta(4\tau)^2} \text{ where } \tau=\frac{it}{2\pi}$$
The other variables appear to involve spherical harmonic oscillators, and we have computed the partition functions for these. If these oscillators are independent, which is 
a wild possibility, then the full partition function is
$$trace(e^{-tH(1)})=\frac{\left(e^{t}\prod_{k=1}^{\infty}\sum_{m,n=0}\mathbf m(m,n)
e^{-t\lambda_{k,m,n}}\right)^2}{\eta(e^{-4t})^2}$$
where 
$$\lambda_{k,m,n}=m^2+m+m(n+\sqrt{n^2+16k^2})+(n+1)\frac{n+\sqrt{n^2+16k^2}}{2}-2k$$
and $\mathbf m(m,n)=1$ when $n=0$ and $=2$ otherwise.

We have written this speculative formula for the Hamiltonian as $H(1)$ because we are assuming the radius of the circle is one. The Hamiltonian
should depend on the radius - since conformal invariance is broken - which we do not fully understand (hence we do not know how to compare the spectrum of $H(1)$ with the conjectural scattering
solution in \cite{ORW}). When the radius is small the conjectured asymptotic freedom of the model should mean the theory behaves like a massless free field, i.e.
that the $\chi$ degrees of freedom, which are Gaussian, should dominate. When the radius is large, i.e. when we are considering the model in the plane, then
the other variables should dominate, and we should be able to use the spherical oscillators to reproduce other expected properties of the model.

In any event we now have a possible formula for the partition function, and we are seeking a numerical test for it.

\end{document}